\documentclass{amsart}

\allowdisplaybreaks
\usepackage{amsbsy}
\usepackage{amsaddr}
\usepackage{luatex85}
\usepackage{color,soul}

\usepackage{amsmath,amsthm}
\usepackage{mathrsfs}
\usepackage{amscd,amssymb, amsfonts, verbatim,subfigure, enumerate}
\usepackage[mathcal]{eucal}
\usepackage[super]{nth}
\usepackage[pdftex]{graphicx}
\usepackage{mathpazo}

\usepackage{color,slashed}
\newcommand{\g}{\mathfrak{g}}

\newcommand{\br}{\overline}

\newcommand{\C}{\mathbb C}

\newcommand{\Z}{\mathbb Z}

\newcommand{\mbb}{\mathbb}

\newcommand{\R}{\mbb R}

\newcommand{\dbar}{\br{\partial}}

 \DeclareMathOperator{\Hom}{Hom}

\DeclareMathOperator{\Tr}{Tr}

\newtheoremstyle{thm}% name
  {7pt}%      Space above
  {7pt}%      Space below
  {\itshape}%         Body font
  {}%         Indent amount (empty = no indent, \parindent = para indent)
  {\bf}% Thm head font
  {.}%        Punctuation after thm head
  {5pt}%     Space after thm head: " " = normal interword space;
  {\thmnumber{#2 }\thmname{#1}\thmnote{ (#3)}}%         Thm head spec (can be left empty, meaning `normal')

\newtheoremstyle{def}% name
  {7pt}%      Space above
  {10pt}%      Space below
  {\itshape}%         Body font
  {}%         Indent amount (empty = no indent, \parindent = para indent)
  {\bf}% Thm head font
  {.}%        Punctuation after thm head
  {5pt}%     Space after thm head: " " = normal interword space;
  {\thmnumber{#2} \thmname{#1}\thmnote{ (#3)}}%         Thm head spec (can be left empty, meaning `normal')

\newtheoremstyle{rem}% name
  {4pt}%      Space above
  {10pt}%      Space below
  {}%         Body font
  {}%         Indent amount (empty = no indent, \parindent = para indent)
  {\itshape}% Thm head font
  {:}%        Punctuation after thm head
  {3pt}%     Space after thm head: " " = normal interword space;
  {}%         Thm head spec (can be left empty, meaning `normal')

\newtheoremstyle{texttheorem}% name
  {8pt}%      Space above
  {8pt}%      Space below
  {\itshape}%         Body font
  {}%         Indent amount (empty = no indent, \parindent = para indent)
  {\bf}% Thm head font
  {. \hspace{5pt}}%        Punctuation after thm head
  {3pt}%     Space after thm head: " " = normal interword space;
  {}%         Thm head spec (can be left empty, meaning `normal')

\newtheorem*{theorem*}{Theorem}
\newtheorem*{lemma*}{Lemma}
\newtheorem*{corollary*}{Corollary}
\newtheorem*{proposition*}{Proposition}
\newtheorem*{definition*}{Definition}

\newtheorem{theorem}{Theorem}[subsection]
\newtheorem{thm-def}{Theorem/Definition}[theorem]
\newtheorem{proposition}[theorem]{Proposition}

\numberwithin{equation}{subsection}
\usepackage{tikz}
\usepackage{quiver}

\theoremstyle{def}

\usepackage[compat=1.1.0]{tikz-feynman} %Feynman diagrams

\theoremstyle{rem}
\usepackage{hyperref}
\usetikzlibrary{shapes.geometric}

\newcommand{\PT}{\mbb{PT}}
\newcommand{\CP}{\mbb{CP}}

\newcommand*\circled[1]{\tikz[baseline=(char.base)]{
            \node[shape=circle,draw,inner sep=2pt] (char) {#1};}}

\title{Twisted holography from the B-model on a 7-fold}
\author{Seraphim Jarov}

\address{Perimeter Institute for Theoretical Physics}
\email{sjarov@perimeterinstitute.ca}

\begin{document}

\begin{abstract}
    I study a topological string construction of the holographic duality between Kodaira-Spencer gravity on the Calabi-Yau 7-fold $\mathcal{O}(-1)^4\to\PT$ in the presence of a stack of $N$ backreacted D5 branes wrapping twistor space, $\PT$. The theory on the stack of branes is the twistor uplift of self-dual $\mathcal{N}=4$ gauge theory. I show that turning on a bulk superpotential and twisting the brane theory by the dual supercharge reduces the duality to twisted holography which relates the B-model on AdS$_3\times S^3\cong SL(2,\C)$ to the 2d chiral algebra subsector of $\mathcal{N}=4$. I do an analogous computation for the twistor uplift of self-dual $\mathcal{N}=2$ by working on the Calabi-Yau 7-fold $\mathcal{O}(-2,-2)\oplus \mathcal{O}(0,-1)^2\to\CP^1\times\PT$. I also connect twists of the twistor uplift of self-dual $\mathcal{N}=4$ with the matrix model found by supersymmetric localization on $S^4$ and the Dijkgraaf-Vafa matrix model construction.
\end{abstract}

\maketitle

\tableofcontents

\section{Introduction}

This paper studies conjectural holographic dualities involving twistor uplifts of certain supersymmetric self-dual theories. In particular, I consider the following boundary theories: the twistor uplifts of self-dual $\mathcal{N}=4$ and $\mathcal{N}=2$ super Yang-Mills (SYM). Motivated by the Witten-Berkovitz twistor string \cite{Witten_2004,Berkovits_2004}, the self-dual $\mathcal{N}=4$ duality was originally proposed in \cite{bittleston2025selfdualgaugetheory}. General constructions of supersymmetric theories on twistor superspaces are discussed in detail in \cite{costello2013notessupersymmetricholomorphicfield}.

The story takes inspiration from the twisted holography program \cite{costello2021twistedholography,Costello:2016mgj}, which was also studied in earlier examples \cite{bonetti2016supersymmetriclocalizationads5protected,Ishtiaque_2020}. Recently, the idea of studying holography for topological strings has been extended to many examples \cite{Costello:2023hmi,bittleston2025selfdualgaugetheory,Costello_2021,Sharma:2025ntb}. There has been particular interest in engineering Calabi-Yaus involving twistor spaces as a way to engineer dualities involving certain 4-dimensional self-dual theories \cite{Costello:2021bah,bittleston2025selfdualgaugetheory,Jarov:2025bmq,Sharma:2025ntb}.

In this work, as suggested in \cite{bittleston2025selfdualgaugetheory}, we consider the B-model topological string on different Calabi-Yau 7-folds to get a higher-dimensional duality than the original construction of \cite{costello2021twistedholography}. For the self-dual $\mathcal{N}=4$ case, we study the B-model on $\mathcal{O}(-1)^4\to\PT$ in the presence of $N$ backreacted D5 branes wrapping the zero section. The self-dual $\mathcal{N}=2$ case arises when we study $\mathcal{O}(0,-1)^2\oplus\mathcal{O}(-2,-2)\to(\CP^1\times\PT)$ where again, we wrap the zero section with $N$ D7 branes. This geometry can also be understood as the $T^*\CP^1\oplus\C^2\to\PT$ associated to the $\C^*$ action scaling the $T^*\CP^1$ fibres with weight $-2$ and each of the $\C$ fibres with weight $-1$.

In both cases, I compute the closed string field sourced by backreacting the $N$ D5/D7 branes. The main finding of this paper is that we can find the original twisted holography duality when we twist the bulk/boundary theories by carefully chosen superpotentials/supercharges. This finding is illustrated in Fig. \ref{fig:overview} for the self-dual $\mathcal{N}=4$ case. In particular, we find that the backreacted bulk geometry twists to $SL(2,\C)$ in the first example and in the second we find $SL(2,\C)/\Z_2$. The theories on the boundaries twist to the 2d chiral algebra subsectors of $\mathcal{N}=4$ and $\mathcal{N}=2$ SYM which were originally found in \cite{Beem_2015}.

\begin{figure}[h]
    \centering
    \includegraphics[width=0.75\linewidth]{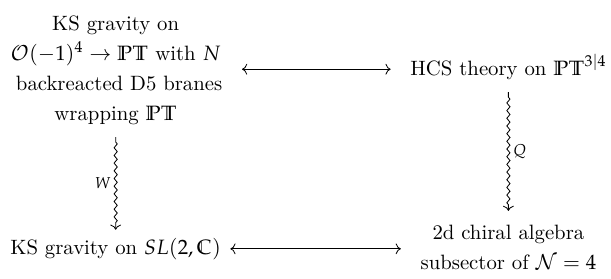}
    \caption{The holographic duality studied in this paper at the top of the figure, where KS stands for Kodaira-Spencer and hCS stands for holomorphic Chern-Simons. The vertical arrows denote twisting by the superpotential/supercharge $W$/$Q$. The bottom row is the duality studied in twisted holography \cite{costello2021twistedholography}.}
    \label{fig:overview}
\end{figure}

At the end of this paper, I connect matrix model subsectors in the theories studied here with known matrix model constructions. I show that the Gaussian matrix model \cite{gopakumar2011simplestgaugestringduality,Gopakumar_2013}  subsector of $\mathcal{N}=4$ SYM found at the poles of $S^4$ via supersymmetric localization \cite{Pestun_2012} can be viewed as a twist of the twistor uplift of self-dual $\mathcal{N}=4$. I also show that the Gaussian matrix model found by \cite{Budzik_2023} living in the chiral algebra subsector of $\mathcal{N}=4$ can be derived from the Dijkgraaf-Vafa matrix model construction \cite{Dijkgraaf_2002}.

\subsection{Future work}

It would be interesting to see how this holographic construction fits into the AdS$_5$/CFT$_4$ correspondence \cite{Maldacena_1999,witten1998antisitterspaceholography} since the theory living on the brane Penrose transforms to the self-dual subsector of $\mathcal{N}=4$. In particular, we would expect to find AdS$_5$ in the backreacted bulk geometry. While I currently do not know how this can emerge, the fact that we can find AdS$_3$ upon localizing provides some evidence that the bulk does contain AdS$_5$. Very recently, an attempt at understanding the flux sourced by the backreaction was made in \cite{Sharma:2025ntb}.

It would be nice to perform thorough tests of the duality. Recently, self-dual $\mathcal{N}=4$ determinants have been studied \cite{Sharma:2025ntb} following the style of checks done in the original twisted holography setting in \cite{Budzik_2023,Lopez-Raven:2024vop}. In section \ref{sec:MM}, I find the matrix model that Pestun uses to compute equatorial Wilson loops in $\mathcal{N}=4$ SYM on $S^4$ \cite{Pestun_2012}. It would be interesting to find a bulk description of these Wilson loops.

In subsection \ref{sec:MM}, I show that a certain Beltrami differential deforms $\C^2/\Z_2\times\C$ (equivalently, on $T^*\CP^1\times\C$ by blowing up the $\C^2/\Z^2$ singularity) to $SL(2,\C)/\Z_2$. It would be interesting if one could find $SL(2,\C)/\Z_k$ when studying the B-model on more general orbifolds $\C^2/\Z_k\times\C$. These orbifold constructions were discussed in \cite{Gaiotto:2024dwr,Abajian:2024rjq}.

\subsection*{Acknowledgments}

Many thanks go to my PhD advisor, Kevin Costello, for suggesting this project and his mentorship. I thank Roland Bittleston for his thoughtful guidance at every stage of this project. I would also like to thank Kasia Budzik and Davide Gaiotto for useful conversations and comments. Research at Perimeter Institute is supported in part by the Government of Canada through the Department of Innovation, Science and Economic Development and by the Province of Ontario through the Ministry of Colleges and Universities.

\section{Background}\label{sec:bgr}

The open string sector of the B-model topological string can be defined on any odd complex-dimensional Calabi-Yau manifold. This was studied on $\C^3$ in twisted holography \cite{costello2021twistedholography}. But an equivalent construction can be made by starting with the resolved conifold
\begin{align}\label{eq:res_coni}
   X= \mathcal{O}(-1)^2\to\CP^1.
\end{align}
The twisted holography correspondence between the 2d chiral algebra subsector of $\mathcal{N}=4$ SYM on the stack of D1 branes and Kodaira-Spencer theory on the bulk $SL(2,\C)$ geometry arises from backreacting $N$ D1 branes wrapping $\CP^1\subset X$.

This paper is interested in understanding topological string constructions of holography on higher-dimensional Calabi-Yau manifolds. The story will be similar to the original twisted holography construction, so I will review some of the basic calculations from their setup. In the examples I study, the dual brane theories are certain twistor uplifts of SYM. So, I will quickly introduce these theories at the end of this section as well. 

\subsection{The backreaction}\label{subsec:bgr-backreact}

Let me review the backreaction computation of \cite{costello2021twistedholography} to see how the $SL(2,\C)$ geometry emerges. I will employ similar techniques in section \ref{sec:N=2} to discover an $SL(2,\C)/\Z_2$ geometry from backreaction on a different background. In section \ref{sec:N=4}, I will also show that the $SL(2,\C)$ geometry emerges from the localization of a backreaction in 7-dimensions. 

The bulk closed string theory of the B-model is known as Kodaira-Spencer gravity \cite{Bershadsky:1993cx,Bershadsky_1993}. The equation of motion for closed string fields is 
\begin{align*}
    \dbar\alpha+\frac{1}{2}[\alpha,\alpha] = 0
\end{align*}
where $[\cdot,\cdot]$ is the Schouten bracket and $\alpha\in PV^{\bullet,\bullet}(X)$ is a polyvector field on the background Calabi-Yau $X$. In twisted holography, $X=\C^3$ (or equivalently, $\mathcal{O}(-1)^2\to\CP^1$). Backreacting a stack of $N$ D1 branes wrapping $\C\subset X$, introduces a source to the equations of motion 
\begin{align*}
    \dbar\alpha+\frac{1}{2}[\alpha,\alpha] + N\delta_{\C}=0.
\end{align*}
Where the $\delta$-function is supported on the stack of branes. If we call the transverse directions $w_1,w_2$ and the brane direction $z$, then a solution to the sourced equation is known as the Bochner-Martinelli kernel
\begin{align*}
    \beta = -\frac{N}{(2\pi)^4}\frac{\epsilon_{ij}\bar{w}^id\bar{w}^j}{||w||^4} \partial_{z}.
\end{align*}
This closed string field is known as a Beltrami differential, which deforms the $\C^3$ complex structure by $\dbar\mapsto\dbar+\beta$. One can then check that the ring of $(\dbar+\beta)$-holomorphic functions is spanned by 4 functions $f_i$ that satisfy the $SL(2,\C)$ condition $f_1f_4-f_2f_3 = N$. For an explicit derivation of these functions, see section 4 of \cite{costello2021twistedholography}.

\subsection{The theory on the brane}\label{subsec:bgr-brane}

To understand the open string sector of the B-model, I will explain the simple example of the theory living on the stack of D1 branes in twisted holography. Let's view $\C^3$ as a trivial $\C^2$ bundle over $\C$ with base coordinate $z$ and fibre coordinates $w_i$ for $i = 1,2$. Then wrapping the locus $w_i = 0$ with $N$ D1 branes leads to the following theory on the brane
\begin{align*}
    \int_{\C^{1|2}}dz\wedge d^2\theta_i\wedge\text{hCS}(\mathcal{A}).
\end{align*}
Where the fibres have been parity shifted and we call the fermionic directions $\theta_i$. The coordinates $\theta_i$ in the brane action can be thought of as T-dual to the $w_i$ bulk coordinates. We study the brane theory on this odd background here because studying topological strings on an even vector bundle $V$ is equivalent to working on the odd vector bundle $\Pi V^\vee$ \cite{coates2001quantumriemannroch,BEILINSON1988317}. This should not be confused with the deeper result coming from twisted holography that the brane theory is equivalent to the bulk theory on the backreacted geometry.

The gauge field is then $\mathcal{A}\in\Omega^{0,\bullet}(\C,\g)[\theta_1,\theta_2]$. If one integrates out the fermionic directions, one lands on the gauged $\beta\gamma$ system studied in \cite{costello2021twistedholography,Beem_2015}. I work this out explicitly in a different, but similar setting in appendix \ref{app:MM-Beltrami}. I will study open string actions in the same fashion on higher-dimensional branes in sections \ref{sec:N=4} and \ref{sec:N=2}.

\subsection{The twistor string}\label{subsec:bgr-twist_string}

In sections \ref{sec:N=4} and \ref{sec:N=2}, I study open strings on D5 and D7 branes, which are equivalent to twistor uplifts of certain supersymmetric self-dual theories. Here I will review the twistor description of self-dual $\mathcal{N}=4$ and how it can be realized as the theory on a D5 brane using the language of this paper.

In \cite{Witten_2004,Berkovits_2004,Boels:2006ir}, holomorphic Chern-Simons theory on twistor superspace $\PT^{3|4}$ was shown to Penrose transform to self-dual $\mathcal{N}=4$ SYM on $\R^4$. In the style of this paper, twistor superspace can be viewed as the total space of 
\begin{align*}
    \Pi\mathcal{O}(1)^4\to\PT.
\end{align*}
Let's call the fermionic directions $\theta^i$, the twistor $\CP^1$ direction $z$, and the twistor fibres $v_{\dot{\alpha}}$. Then the twistor action for self-dual $\mathcal{N}=4$ is 
\begin{align}\label{eq:twistorstring}
    \int_{\Pi\mathcal{O}(1)^4\to\PT}\Omega\wedge d^4\theta^i\wedge \text{hCS}(\mathcal{A}).
\end{align}
Where $\Omega$ is the meromorphic volume form on $\PT$ and our superfield is valued in 
\begin{align*}
    \mathcal{A}\in\Omega^{0,1}(\PT,\g)[\theta^i]_{i=1,2,3,4}.
\end{align*}
This is the theory studied in \cite{Witten_2004,Berkovits_2004,Boels:2006ir}. 

\textit{How does this theory arise from the topological string?} Consider the even holomorphic vector bundle
\begin{align*}
    X=\mathcal{O}(-1)^4\to\PT.
\end{align*}
As stated in subsection \ref{subsec:bgr-brane}, the theory on a stack of $N$ D5 branes wrapping $\PT\subset X$ is holomorphic Chern-Simons theory on the odd vector bundle $\Pi(\mathcal{O}(-1)^4)^{\check{}}=\Pi\mathcal{O}(1)^4$. So, we see that studying the B-model on $X$ with D5 branes gives us a topological string construction of the twistor uplift of self-dual $\mathcal{N}=4$.

\section{The self-dual $\mathcal{N}=4$ cosntruction}\label{sec:N=4}

Consider the Calabi-Yau 7-fold
\begin{align*}
    X = \mathcal{O}(-1)^4\oplus\mathcal{O}(1)^2\to\CP^1.
\end{align*}
Notice that this is in some sense a higher-dimensional upgrade of the resolved conifold introduced in equation \ref{eq:res_coni}. Also note that this is the same as $\mathcal{O}(-1)^4\to\PT$. Let's call the $\mathcal{O}(-1)^4$ directions $w_i$ with $i = 1,2,3,4$, the $\mathcal{O}(1)^2$ directions $v_{\dot{\alpha}}$ with $\dot{\alpha}=\dot{1},\dot{2}$, and the $\CP^1$ coordinate $z$.

It is easy to see that $X$ is Calabi-Yau since the canonical bundle of $\CP^1$ is $K_{\CP^1}=\mathcal{O}(2)$.

\subsection{The backreaction}\label{subsec:N=4BR}

To build a holographic duality, we wrap $\PT\subset X$ with $N$ D5 branes and compute the backreaction. As discussed in the background section, this amounts to solving the Kodaira-Spencer equation of motion in the presence of a source term
\begin{align}\label{eq:breact_ODE}
    \dbar\alpha+\frac{1}{2}[\alpha,\alpha] + N\delta_\PT=0.
\end{align}
We will soon see that our choice of Beltrami differential has no $z$ dependence on the term containing $\partial_z$. So, we are interested in solving
\begin{align*}
    \dbar\alpha=-N\delta_{\PT}.
\end{align*}
As in the 3-dimensional case, the solution is the 4-dimensional Bochner-Martinelli kernel, 
\begin{align*}
    \alpha = -\frac{3!N}{(2\pi)^4}\frac{\epsilon_{ijkl}\bar{w}^id\bar{w}^jd\bar{w}^kd\bar{w}^l}{||w||^8} \partial_z\partial_{v_{\dot{1}}}\partial_{v_{\dot{2}}}.
\end{align*}
In the 3-dimensional case, the backreaction sourced a deformation of the complex structure which deformed the theory on $\mathcal{O}(-1)^2\to\CP^1$ to a theory on $SL(2,\C)$. In this case, however the backreaction sources some higher form polyvector field which does not have a natural interpretation as a deformation to the bulk geometry.

We will see later that the theory on the brane is the twistor uplift of self-dual $\mathcal{N}=4$ SYM, so we do expect to find some subsector of ordinary AdS$_5$/CFT$_4$ from our construction. In particular, we would expect that the flux sourced by the backreaction relates the bulk geometry to AdS$_5$ in some way. While finding such a relation is beyond the scope of this paper, we will see that turning on a superpotential in the bulk localizes the backreacted geometry to AdS$_3\times S^3$.

\subsection*{Turning on a superpotential}

Here, I connect the above bulk theory to the twisted holography program by localizing to the zero locus of a chosen superpotential. This also provides basic evidence that the bulk geometry may contain AdS$_5$ since we find AdS$_3\times S^3\cong SL(2,\C)$ upon localizing.

To connect to twisted holography, we want to localize to $\mathcal{O}(-1)^2\to\CP^1$, so it is natural for us to consider the following superpotential
\begin{align*}
    W = v_{\dot{1}}w_3+v_{\dot{2}}w_4.
\end{align*}
The localization amounts to solving the modified equation of motion
\begin{align}\label{eq:sup_pot_eq}
    \dbar\alpha+\frac{1}{2}[\alpha,\alpha]+[W,\alpha] + N\delta_\PT=0.
\end{align}
% Because we just did $\alpha\to\alpha+W$
Again, we look for solutions with no dependence on the vector field directions, so we can ignore the $[\alpha,\alpha]$ term when solving. Our approach will be to solve this order by order to find that the solution takes the form $\alpha=\delta\alpha+\delta^{(2)}\alpha+\delta^{(3)}\alpha$. At leading order, we have
\begin{align*}
    \delta\alpha = -\frac{3!N}{(2\pi)^4}\frac{\epsilon_{ijkl}\bar{w}^id\bar{w}^jd\bar{w}^kd\bar{w}^l}{||w||^8} \partial_z\partial_{v_{\dot{1}}}\partial_{v_{\dot{2}}}.
\end{align*}
Which is the solution to the backreaction without the superpotential. Plugging this into our expression for $\alpha$, we then get the equation we need to solve
\begin{align*}
    \dbar \delta^{(2)}\alpha+\dbar\delta^{(3)}\alpha+[W,\delta\alpha]+[W,\delta^{(2)}\alpha]+[W,\delta^{(3)}\alpha] = 0.
\end{align*}
The second and third-order terms will be shown to satisfy
\begin{align}\label{eq:higher_order_ODE}
    \dbar \delta^{(2)}\alpha+[W,\delta\alpha]=0\qquad \dbar\delta^{(3)}\alpha+[W,\delta^{(2)}\alpha]=0\qquad [W,\delta^{(3)}\alpha]=0.
\end{align}
Hence implying that $\alpha$ solves equation \ref{eq:sup_pot_eq}. The solutions at each order are the following
\begin{align*}
    \delta^{(2)}\alpha=\frac{3!N}{3(2\pi)^4}\left(\frac{\epsilon_{ij3k}\bar{w}^id\bar{w}^jd\bar{w}^k}{3||w||^6}\partial_{z}\partial_{v_{\dot{2}}}+\frac{\epsilon_{ijk4}\bar{w}^id\bar{w}^jd\bar{w}^k}{||w||^6}\partial_{z}\partial_{v_{\dot{1}}}\right)\qquad \delta^{(3)}\alpha= -\frac{3!N}{3(2\pi)^4}\frac{\epsilon_{ij34}\bar{w}^id\bar{w}^j}{||w||^4} \partial_{z}.
\end{align*}
I will now show that these solve equation \ref{eq:higher_order_ODE}. Following the calculation we performed in appendix \ref{app:BM_calc}, we find that 
\begin{align*}
    \dbar \delta^{(2)}\alpha & = \frac{3!N}{(2\pi)^4}\frac{\epsilon_{ijkl}\bar{w}^id\bar{w}^jd\bar{w}^kd\bar{w}^l}{||w||^8}\partial_{z}\partial_{v_{\dot{2}}}w_3+\frac{3!N}{(2\pi)^4}\frac{\epsilon_{ijkl}\bar{w}^id\bar{w}^jd\bar{w}^kd\bar{w}^l}{||w||^8}\partial_{z}\partial_{v_{\dot{1}}}w_4.
\end{align*}
The above cancels
\begin{align*}
    [W,\delta\alpha]=-\frac{3!N}{(2\pi)^4}\frac{\epsilon_{ijkl}\bar{w}^id\bar{w}^jd\bar{w}^kd\bar{w}^l}{||w||^8} \partial_z\partial_{v_{\dot{2}}}w_3-\frac{3!N}{(2\pi)^4}\frac{\epsilon_{ijkl}\bar{w}^id\bar{w}^jd\bar{w}^kd\bar{w}^l}{||w||^8} \partial_z\partial_{v_{\dot{1}}}w_4.
\end{align*}
So, we have verified that the second-order correction satisfies the correct differential equation. Next, we see that
\begin{align*}
    \dbar\delta^{(3)}\alpha = -\frac{3!N}{6(2\pi)^4}\frac{\epsilon_{ij3k}\bar{w}^id\bar{w}^jd\bar{w}^k}{||w||^6}w_4-\frac{3!N}{6(2\pi)^4}\frac{\epsilon_{ijk4}\bar{w}^id\bar{w}^jd\bar{w}^k}{||w||^6}w_3
\end{align*}
again, using the result in appendix \ref{app:BM_calc}. This cancels
\begin{align*}
    [W,\delta^{(2)}\alpha] = \frac{3!N}{6(2\pi)^4}\frac{\epsilon_{ij3k}\bar{w}^id\bar{w}^jd\bar{w}^k}{3||w||^6}\partial_{z}w_4+\frac{3!N}{6(2\pi)^4}\frac{\epsilon_{ijk4}\bar{w}^id\bar{w}^jd\bar{w}^k}{||w||^6}\partial_{z}w_3.
\end{align*}
Finally, we need to show that the third-order correction is annihilated by $[W,\cdot]$. This follows from the fact that $\delta^{(3)}\alpha$ has no component along the $\partial_{v_{\dot{\alpha}}}$ vector directions. This shows that the solution to equation \ref{eq:sup_pot_eq} is
\begin{align*}
    \alpha =2\frac{3!N}{(2\pi)^4}\left(-\frac{\epsilon_{ijkl}\bar{w}^id\bar{w}^jd\bar{w}^kd\bar{w}^l}{||w||^8} \partial_{v_{\dot{1}}}\partial_{v_{\dot{2}}}+\frac{\epsilon_{ij3k}\bar{w}^id\bar{w}^jd\bar{w}^k}{6||w||^6}\partial_{v_{\dot{2}}}+\frac{\epsilon_{ijk4}\bar{w}^id\bar{w}^jd\bar{w}^k}{6||w||^6}\partial_{v_{\dot{1}}}-\frac{\epsilon_{ij34}\bar{w}^id\bar{w}^j}{6||w||^4} \right)\partial_{z}.
\end{align*}
This superpotential localizes us to the $w_1,w_2,z$ space, so the Beltrami only sees the third-order correction
\begin{align*}
    -\frac{2N}{(2\pi)^4}\frac{\epsilon_{ij}\bar{w}^id\bar{w}^j}{||w||^4} \partial_{z}
\end{align*}
where $w_i=1,2$. This is the Beltrami differential that deformed the resolved conifold to $SL(2,\C)$ in twisted holography, which I reviewed in the background section \ref{sec:bgr}. We have thus found AdS$_3\times S^3$ when we localize the backreacted 7-fold!

\subsection{The theory on the brane}\label{subsec:N=4BRANE}

As explained in the background section \ref{sec:bgr}, the theory on a stack of D5 branes is defined on the super twistor space $\PT^{3|4}$.

I also explained that the open string action on this stack of D5 branes is equivalent to the twistor uplift of self-dual $\mathcal{N}=4$.

Since we are studying holography, we expect to find a dual computation to the localization we performed in the bulk. Indeed, we consider twisting the theory on the stack of D5 branes by the dual supercharge
\begin{align*}
    Q = v_{\dot{1}}\partial_{\theta_3}+v_{\dot{2}}\partial_{\theta_4}.
\end{align*}
Where $\theta_i$ are the coordinates in the fermionic directions. I will show that upon twisting by $Q$, the theory on the D5 branes becomes the boundary theory studied in twisted holography: the chiral algebra subsector of $\mathcal{N}=4$ SYM. Along with our study of the backreaction, this shows that studying the topological string on $\mathcal{O}(-1)^4\oplus\mathcal{O}(1)^2\to\CP^1$ introduces a new holographic duality which encompasses the original story on $\mathcal{O}(-1)^2\to\CP^1$.

\subsection*{The twisted theory}

Twisting open strings in the B-model is an easy procedure, and this example should be straightforward for most familiar with topological strings. Nevertheless, I outline the argument here for completeness.

The holomorphic Chern-Simons action on the stack of D5 branes was presented in equation \ref{eq:twistorstring}. Recall that the gauge field was valued in
\begin{align*}
    \mathcal{A}\in\Omega^{0,1}(\PT,\g)[\theta^i]_{i=1,2,3,4}.
\end{align*}
The theory twists to a theory on the zero locus $Z(Q)\subset\PT$. We will see that in $Q$-cohomology, the gauge field becomes a gauge field valued in $\Omega^{0,1}(\CP^1,\g)[\theta^1,\theta^2]$. Holomorphic Chern-Simons with this field content on a D1 brane wrapping $\CP^1\subset \mathcal{O}(-1)^2\to\CP^1$ matches the 2d chiral algebra studied in the twisted holography setup of \cite{costello2021twistedholography}.

Holomorphic Chern-Simons theory is a theory of $(0,\bullet)$-forms valued in a dg algebra. When we twist, this complex is the Koszul resolution of the locus $v^{\dot{\alpha}}=0$ for $\dot{\alpha}=\dot{1},\dot{2}$.

If we call the Koszul resolution $C$, then it is enough to find an isomorphism of Dolbeault complexes
    \begin{align*}
        \Omega^{0,\bullet}(\PT,C)\cong \Omega^{0,\bullet}(\CP^1)[\theta_1,\theta_2].
    \end{align*}
Consider the complex $B$ given by
\begin{align*}
    0\to\mathcal{O}_{Z(Q)}[\theta_1,\theta_2].
\end{align*}
Then we have an isomorphism of complexes $C\to B$ given by the zero map everywhere except in degree 0 where we have the quotient map $\mathcal{O}_\PT[\theta_1,\theta_2]\to\mathcal{O}_\PT[\theta_1,\theta_2]/Z(Q)\cong \mathcal{O}_{Z(Q)}[\theta_1,\theta_2]$. This induces the isomorphism we want on Dolbeault complexes. So, we see that the twistor uplift of self-dual $\mathcal{N}=4$ twists to the chiral algebra subsector of $\mathcal{N}=4$.

\section{The self-dual $\mathcal{N}=2$ construction}\label{sec:N=2}

I will now introduce a Calabi-Yau 7-fold, which will lead us to a holographic duality with self-dual $\mathcal{N}=2$ on the brane. The geometry we study is the following
\begin{align}\label{eq:N=2_CY7}
    X = \mathcal{O}(0,-1)^2\oplus \mathcal{O}(-2,-2)\to(\CP^1\times\PT)
\end{align}
where my notation on the $\mathcal{O}(\cdot,\cdot)$ factors denotes the projective weight on the $\CP^1$ in the first entry and the second entry is the weight on $\PT$. This is easily seen to be Calabi-Yau since $K_{\CP^1}=\mathcal{O}(2)$ and $K_\PT=\mathcal{O}(4)$.

Let me provide context on the origin of this geometry. In the self-dual $\mathcal{N}=4$ case, we promoted the resolved conifold $\mathcal{O}(-1)^2\to\CP^1$ studied in twisted holography to the CY7 $\mathcal{O}(-1)^4\oplus\mathcal{O}(1)^2\to\CP^1$. In essence, we added 4 more directions so that we would have a Calabi-Yau fibration over twistor space. Here we want to play the same game with the orbifold $\C^2/\Z_2\times\C$. The resolution of the $A_1$ singularity is $T^*\CP^1\times\C$. In promoting this CY3 to a CY7 fibering over twistor space, we land on the geometry shown in equation \ref{eq:N=2_CY7}.

In this setting, we wrap stacks of $N$ D7 branes/anti-branes around $\CP^1\times\PT\subset X$. In the following subsections, I will compute the backreaction and show that it localizes to $SL(2,\C)/\Z_2$ when we introduce the appropriate superpotential. I also describe the theory on the stack of D7 branes and explain why it is the twistor uplift of self-dual $\mathcal{N}=2$ SYM. To wrap up the section, I will state the theory on the brane and comment on it's twist by the dual supercharge to the bulk superpotential.

Throughout the rest of the section, I will denote the twistor coordinates with $v_{\dot{\alpha}},z$ as in the previous section, the second $\CP^1$ factor will have the coordinate $u$, the $\mathcal{O}(-2,-2)$ will be denoted $w_3$, and the $\mathcal{O}(0,-1)^2$ directions will be denoted $w_1,w_2$.

\subsection{The brane construction}

Engineering the brane configurations in this setting is a slightly more delicate task than in the $\mathcal{N}=4$ case. Their construction also affects the backreaction computation in a nontrivial manner. I will first discuss the construction and brane theory here.

Working on the resolution of an $A_1$ singularity, we study a stack of \textit{fractional} D5 branes, which amounts to studying sheaves on the zero section $\CP^1\times\PT\subset X$. The theory we will find on the stack of branes is a holomorphic $\mathcal{N}=2$ quiver gauge theory. The sheaves we study are the following: a stack of D7 branes carrying $N$ copies of the trivial bundle $\mathcal{O}^N$ and a stack of D7 anti-branes carrying the parity shifted bundle $\Pi\mathcal{O}(-1,1)^N$.

This brane configuration leads us to studying two copies of self-dual $\mathcal{N}=2$ gauge theory described by the following quiver

\[\begin{tikzcd}
	\circled{$N$} && \circled{$N$}
	\arrow[curve={height=-18pt}, from=1-1, to=1-3]
	\arrow[curve={height=18pt}, from=1-1, to=1-3]
\end{tikzcd}\]

The origin of the two copies of pure self-dual $\mathcal{N}=2$ gauge theory is explained in subsection \ref{subsec:N=2BRANE}. They come from the twistor construction presented in \cite{Boels:2006ir,costello2013notessupersymmetricholomorphicfield}. Let me now describe the matter fields. They come from the D7-D7 strings propagating between the branes and antibranes. Using T-duality, the theory on the branes involves fields valued in $\Omega^{0,\bullet}(\CP^1\times\PT,\g)[\theta_1,\theta_2,\zeta]$ where $\theta_i\in \Pi\mathcal{O}(0,1)^2$ and $\zeta \in \Pi\mathcal{O}(2,2)$. Then the theory on the stack of branes has fields accompanying $1,\theta_i,\theta^2,\theta_i\zeta,\theta^2\zeta$, whereas the anti-branes has fields with projective weight shifted by 1, so they accompany $\theta_i,\theta^2,\theta_i\zeta$. The resulting hypermultiplet we find is
\begin{align*}
    \Hom(\C^N,\C^N)\otimes (\mathcal{O}(-1)\oplus\mathcal{O}(-2)^2\oplus\mathcal{O}(-3))\oplus \Hom(\C^N,\C^N)\otimes (\mathcal{O}(-1)\oplus\mathcal{O}(-2)^2\oplus\mathcal{O}(-3)).
\end{align*}
Abstractly, the above comes from computing
\begin{align*}
    \Hom(\Pi\mathcal{O}_{\CP^1\times\PT}(-1,1)^N,\mathcal{O}_{\CP^1\times\PT}^N)=\Hom(\C^N,\C^N)\otimes H(\CP^1,\C[\theta_1,\theta_2,\zeta]\otimes\mathcal{O}(1,-1))
\end{align*}
by noticing that the factors accompanying $\zeta\in \mathcal{O}(-1,-3)$ and $\theta^2\zeta\in\mathcal{O}(-1,-5)$ (note that the coordinates are twisted since we tensored with $\mathcal{O}(1,-1)$) vanish since $H^\bullet (\CP^1,\mathcal{O}(-1)) = 0$.

\subsection{The backreaction}\label{subsec:N=2BR}

As explained in the background section \ref{sec:bgr}, we compute the backreaction by first solving the Kodaira-Spencer equation in the presence of a source term. In this setting, however, we have to be careful when writing down the source term. The source term comes from the Chern class of the brane $\text{ch}(\mathcal{O}^N\oplus\Pi\mathcal{O}(-1,1)^N)$ multiplied by the delta function supported on the brane. In our case, we find that
\begin{align*}
    \text{ch}_1(\mathcal{O}^N\oplus\Pi\mathcal{O}(-1,1)^N) = N\frac{dud\bar{u}}{(1+|u|^2)^2}-N\frac{dzd\bar{z}}{(1+|z|^2)^2}.
\end{align*}
and the rest vanish. So, we need to solve
\begin{align*}
    \dbar\alpha+\frac{1}{2}[\alpha,\alpha]+ N\delta_{w_i=0}\frac{dud\bar{u}}{(1+|u|^2)^2}-N\delta_{w_i=0}\frac{dzd\bar{z}}{(1+|z|^2)^2}=0.
\end{align*}
Solving the above is less trivial than the $\mathcal{N}=4$ case because inserting a Bochner-Martinelli type kernel to solve
\begin{align*}
    \dbar\alpha'+ N\delta_{w_i=0}\frac{dud\bar{u}}{(1+|u|^2)^2}=0\quad\text{and}\quad \dbar\alpha''- N\delta_{w_i=0}\frac{dzd\bar{z}}{(1+|z|^2)^2}=0
\end{align*}
separately potentially introduces counterterms from the Schouten bracket. Solving the above separately gives the polyvectors 
\begin{align*}
    \alpha' = -i\frac{2!N}{(2\pi)^3}\frac{\epsilon^{ijk}\bar{w}_id\bar{w}_jd\bar{w}_k}{||w||^6}\frac{d^3wdud\bar{u}}{(1+|u|^2)^2}\qquad \alpha'' = -i\frac{2!N}{(2\pi)^3}\frac{\epsilon^{ijk}\bar{w}_id\bar{w}_jd\bar{w}_k}{||w||^6}\frac{d^3wdzd\bar{z}}{(1+|z|^2)^2}
\end{align*}
Using the volume form $\Omega = u^{-1}dud^3w_idzd^2v_{\dot{\alpha}}$, we can use the isomorphism $PV^{j,i}(X)\cong\Omega^{7-j,i}(X)$ to write our solutions as polyvectors
\begin{align}\label{eq:N=2_BM_ker}
    \alpha'& = -i\frac{2!N}{(2\pi)^3}\frac{\epsilon^{ijk}\bar{w}_id\bar{w}_jd\bar{w}_k}{||w||^6}\frac{ud\bar{u}}{(1+|u|^2)^2}\partial_z\partial_{v_{\dot{1}}}\partial_{v_{\dot{2}}}\\ \nonumber\alpha''& = -i\frac{2!N}{(2\pi)^3}\frac{\epsilon^{ijk}\bar{w}_id\bar{w}_jd\bar{w}_k}{||w||^6}\frac{ud\bar{z}}{(1+|z|^2)^2}\partial_u\partial_{v_{\dot{1}}}\partial_{v_{\dot{2}}}.
\end{align}
Notice that we are technically abusing notation by leaving the names of $\alpha'$ and $\alpha''$ unchanged. Building a solution by summing $\alpha'$ and $\alpha''$ would generate a term of the form $[\alpha',\alpha'']$ in the Maurer-Cartan equation. But one then easily sees that such a term vanishes because it would have $4$ antiholomorphic forms in the $w_i$ directions. So, we see that the backreaction of our stack of branes sources the flux $\alpha = \alpha'+\alpha''$.

\subsection*{Localization}

We are interested in localizing to the critical locus of a superpotential, which we will choose to be
\begin{align*}
    W = zw_1+v_{\dot{2}}w_2.
\end{align*}
This choice discards the $z$-dependence of our solution, which means we can drop $\alpha''$ in the following analysis. Notice that in this setting we are localizing to one of the fibre planes rather than a twistor line as we did in the $\mathcal{N}=4$ case.

Our choice of superpotential localizes us to the Calabi-Yau 3-fold used to study the twisted holography dual of the chiral algebra subsector of $\mathcal{N}=2$ SYM is $T^*\CP^1\times\C$ (or equivalently, $\C^2/\Z_2\times\C$). We then modify the source equation by\footnote{Note that we have dropped the source term depending on $z$ due to our analysis earlier in the section.}
\begin{align*}
    \dbar\alpha +\frac{1}{2}[\alpha,\alpha]+[W,\alpha] + N\delta_{w_i=0}\frac{dud\bar{u}}{(1+|u|^2)^2}=0.
\end{align*}
As in the $\mathcal{N}=4$ story, we will look for a solution of the form $\alpha=\delta\alpha+\delta^{(2)}\alpha+\delta^{(3)}\alpha$ where the first-order contribution is given by the Bochner-Martinelli kernel displayed in equation \ref{eq:N=2_BM_ker}. Plugging $\alpha$ into the modified source equation leads us to consider differential equations of the same form as in equation \ref{eq:higher_order_ODE}.

At second-order, we find the correction
\begin{align*}
    \delta^{(2)}\alpha = i\frac{2!N}{2(2\pi)^3}\frac{\epsilon_{1ij}\bar{w}^id\bar{w}^j}{||w||^4}\frac{ud\bar{u}}{(1+|u|^2)^2}\partial_{v_{\dot{1}}}\partial_{v_{\dot{2}}}+i\frac{2!N}{2(2\pi)^3}\frac{\epsilon_{i2j}\bar{w}^id\bar{w}^j}{||w||^4}\frac{ud\bar{u}}{(1+|u|^2)^2}\partial_{z}\partial_{v_{\dot{1}}}.
\end{align*}
To check that this solves the correct differential equation $\dbar\delta^{(2)}\alpha+[W,\delta\alpha]=0$, we apply the calculations in appendix \ref{app:BM_calc} to find
\begin{align*}
    \dbar\delta^{(2)}\alpha = i\frac{2!N}{(2\pi)^3}\frac{\epsilon^{ijk}\bar{w}_id\bar{w}_jd\bar{w}_k}{||w||^6}\frac{ud\bar{u}}{(1+|u|^2)^2}\partial_{v_{\dot{1}}}\partial_{v_{\dot{2}}}w_1+i\frac{2!N}{(2\pi)^3}\frac{\epsilon^{ijk}\bar{w}_id\bar{w}_jd\bar{w}_k}{||w||^6}\frac{ud\bar{u}}{(1+|u|^2)^2}\partial_z\partial_{v_{\dot{1}}}w_2.
\end{align*}
Which cancels against
\begin{align*}
    [W,\delta\alpha] = -i\frac{2!N}{(2\pi)^3}\frac{\epsilon^{ijk}\bar{w}_id\bar{w}_jd\bar{w}_k}{||w||^6}\frac{ud\bar{u}}{(1+|u|^2)^2}\partial_{v_{\dot{1}}}\partial_{v_{\dot{2}}}w_1-i\frac{2!N}{(2\pi)^3}\frac{\epsilon^{ijk}\bar{w}_id\bar{w}_jd\bar{w}_k}{||w||^6}\frac{ud\bar{u}}{(1+|u|^2)^2}\partial_z\partial_{v_{\dot{1}}}w_2.
\end{align*}
The next correction is given by
\begin{align*}
    \delta^{(3)}\alpha=-i\frac{2!N}{4(2\pi)^3}\frac{\bar{w}^3}{||w||^2}\frac{ud\bar{u}}{(1+|u|^2)^2}\partial_{v_{\dot{1}}}.
\end{align*}
The $\dbar$-operator applied to this gives
\begin{align*}
    \dbar\delta^{(3)}\alpha=-i\frac{2!N}{4(2\pi)^3}\frac{\epsilon_{1ij}\bar{w}^id\bar{w}^j}{||w||^4}\frac{ud\bar{u}}{(1+|u|^2)^2}\partial_{v_{\dot{1}}}w_2-i\frac{2!N}{4(2\pi)^3}\frac{\epsilon_{i2j}\bar{w}^id\bar{w}^j}{||w||^4}\frac{ud\bar{u}}{(1+|u|^2)^2}\partial_{v_{\dot{1}}}w_1.
\end{align*}
It is easy to see that this is equal to $-[W,\delta^{(2)}\alpha]$. The last piece to check is that $[W,\delta^{(3)}\alpha]=0$, but this follows trivially since $\delta^{(3)}\alpha$ has no derivatives in the directions of $W$.

We have thus shown that 
\begin{align*}
    \alpha = i\frac{2\cdot 2!N}{(2\pi)^3}\left(-\frac{\epsilon^{ijk}\bar{w}_id\bar{w}_jd\bar{w}_k}{||w||^6}\partial_z\partial_{v_{\dot{2}}}+\frac{\epsilon_{1ij}\bar{w}^id\bar{w}^j}{4||w||^4}\partial_{v_{\dot{2}}}+\frac{\epsilon_{i2j}\bar{w}^id\bar{w}^j}{4||w||^4}\partial_{z}-\frac{\bar{w}^3}{8||w||^2}\right)\frac{ud\bar{u}}{(1+|u|^2)^2}\partial_{v_{\dot{1}}}
\end{align*}
solves the modified source equation. If we localize to the critical locus of $W$, then we find 
\begin{align}\label{eq:SL/Z2}
    -i\frac{N}{2(2\pi)^3}\frac{\bar{w}^3}{||w||^2}\frac{ud\bar{u}}{(1+|u|^2)^2}\partial_{v_{\dot{1}}}.
\end{align}

\subsection*{Finding $SL(2,\C)/\Z_2$} I claimed earlier that this choice of superpotential would localize our backreacted 7d geometry to the backreaction on $\C^2/\Z_2\times\C$. We know that the backreaction of a stack of $N$ D1 branes wrapping $\C\subset \C^2/\Z_2\times\C$ deforms the complex structure to $SL(2,\C)/\Z_2$. I will now show that the Beltrami differential derived in equation \ref{eq:SL/Z2} gives rise to this geometry.

I follow the same procedure used to find $SL(2,\C)$ from the original twisted holography program \cite{costello2021twistedholography}, which was reviewed in \ref{sec:bgr}. The $\dbar$-operator gets deformed to 
\begin{align*}
    \dbar\to\bar{D}=\dbar+ \frac{N}{2(2\pi)^3}\frac{ud\bar{u}}{t(1+|u|^2)^2}\partial_{v}
\end{align*}
% \hl{we removed a factor of $u$ upstairs} 
where I will now call the $\mathcal{O}(-2)$ coordinate $t$ rather than $w^3$ and I drop the index on $v$. Let's find the holomorphic functions in this complex structure. The projective weight of $t$ tells us that we can have the following holomorphic functions
\begin{align*}
    f_1 = t\qquad f_2 = tu\qquad f_3 = tu^2
\end{align*}
which leads us to construct 3 more functions which are linear in $N$
\begin{gather*}
    f_4 = vt-\frac{N}{2(2\pi)^3}\frac{1+2|u|^2}{(1+u\bar{u})}\qquad f_5 = vtu-\frac{N}{2(2\pi)^3}\frac{u(1+2|u|^2)}{1+u\bar{u}}\\ f_6 = vtu^2-\frac{N}{2(2\pi)^3}\frac{u^2(1+2|u|^2)}{1+u\bar{u}}.
\end{gather*}
It is not hard to see that these are holomorphic in the new complex structure. I check this explicitly in appendix \ref{app:SL(2,C)/Z2calc}.

Next, we want to find holomorphic functions that depend quadratically in $N$. Finding such functions is less obvious. They can be systematically found by applying the $\bar{D}$-operator to $v^2f_1,v^2f_2,v^2f_3$ and solving for correction terms order-by-order in $N$ until one arrives at a holomorphic function. Here, I will give the resulting functions from such a computation and check that they are indeed annihilated by $\bar{D}$. The order $N^2$ functions are
\begin{align*}
    f_7 & = v^2t -\frac{N}{(2\pi)^3}\frac{v(1+2|u|^2)}{(1+u\bar{u})}+\frac{N^2}{8(2\pi)^6}\frac{(1+2|u|^2)^2}{t(1+|u|^2)^2}\\
    f_8 & = v^2tu -\frac{N}{(2\pi)^3}\frac{vu(1+2|u|^2)}{(1+u\bar{u})}+\frac{N^2}{8(2\pi)^6}\frac{u(1+2|u|^2)^2}{t(1+|u|^2)^2}\\ 
    f_9 & = v^2tu^2 -\frac{N}{(2\pi)^3}\frac{vu^2(1+2|u|^2)}{(1+u\bar{u})}+\frac{N^2}{8(2\pi)^6}\frac{u^2(1+2|u|^2)^2}{t(1+|u|^2)^2}
\end{align*}
whose $\bar{D}$-holomorphicity is verified in appendix \ref{app:SL(2,C)/Z2calc}. We have thus found a family of 9 holomorphic functions in the complex structure defined by $\bar{D}$. Just like how we found $SL(2,\C)$ from the backreaction in twisted holography, we will look for relations on the $f_i$ that match the ring of holomorphic functions on $SL(2,\C)/\Z_2$. The relations are the following:
\begin{align}\label{eq:SL/Z2_rel}
    f_2^2 = f_1f_3 \qquad f_8^2 = f_7f_9\qquad f_4^2 = f_1f_7 \qquad f_5^2 = f_3f_7\qquad f_6^2 = f_3f_9\qquad f_4f_2 = f_1f_5 .
\end{align}
I check these explicitly in appendix \ref{app:SL(2,C)/Z2calc}. This then gives a map to the space of $\Z_2$-invariants on $SL(2,\C)/\Z_2$. Indeed, if we consider a general matrix element specified by its entries $a,b,c,d$ subject to the $SL(2,\C)$ relation $ad-bc=N$, then the $\Z_2$ invariants are quadratic terms in these elements. There are ten such elements where one can be removed using $ad-bc=N$. We are thus left with the following identification which is consistent with the relations we found above in equation \ref{eq:SL/Z2_rel}:
\begin{gather*}
    f_1 = a^2 \qquad f_2 = ab \qquad f_3 = b^2 \qquad f_4 = ac \qquad f_5 = bc \\ f_6 = bd \qquad f_7 = c^2\qquad  f_8 = cd\qquad f_9 = d^2.
\end{gather*}

\subsection*{The relation to $SO(3,\C)$} 

I have found an explicit matching between holomorphic functions in the presence of a certain Beltrami differential and $\Z_2$-invariant functions on $SL(2,\C)$. There is another viewpoint we can take here via the isomorphism $SL(2,\C)/
\Z_2\cong SO(3,\C)$.

In particular, we can build a $3\times3$ matrix out of our nine functions $f_i$ in the following manner
\begin{align*}
    R=
\begin{pmatrix}
\dfrac{f_1-f_3-f_7+f_9}{2}
&
\dfrac{i\,(f_1-f_3+f_7-f_9)}{2}
&
-f_4+f_6
\\[10pt]
\dfrac{i\,(-f_1-f_3+f_7+f_9)}{2}
&
\dfrac{f_1+f_3+f_7+f_9}{2}
&
i\,(f_4+f_6)
\\[10pt]
-f_2+f_8
&
i\,(-f_2-f_8)
&
1+2f_5
\end{pmatrix}.
\end{align*}
One can then check that the relations on the functions $f_i$ imply $R^TR=N\mathbb{1}$. This is another way we can see the backreacted geometry from the presence of the Beltrami differential.

\subsection{The theory on the brane}\label{subsec:N=2BRANE}

In the same spirit as the other brane theories studied in this paper, the action takes the form
\begin{align*}
    \int_{\tilde{X}}\Omega\wedge\omega\wedge d^2\theta_i\wedge d\zeta\wedge \text{hCS}(\mathcal{A})
\end{align*}
where 
\begin{align*}
    \tilde{X}=\Pi\mathcal{O}(0,1)^2\oplus\Pi\mathcal{O}(2,2)\to(\CP^1\times\PT)\quad\text{and}\quad \mathcal{A}\in\Omega^{0,1}(\CP^1\times\PT,\g)[\theta_1,\theta_2,\zeta].
\end{align*}
Note that $\Omega$ is the $\PT$ volume form, $\omega$ is the $\CP^1$ volume form, and we have named coordinates $\theta^i\in \Pi\mathcal{O}(1,1)^2$ and $\zeta\in \Pi\mathcal{O}(0,2)$. This is the twistor uplift of self-dual $\mathcal{N}=2$ when we compactify the second $\CP^1$ and integrate out $\zeta$ as discussed in \cite{Boels:2006ir}.

The dual computation of the bulk localization involves twisting by the supercharge
\begin{align*}
    Q=z\partial_{\theta_1}+v_{\dot{2}}\partial_{\theta_2}.
\end{align*}
Repeating the arguments I made for the twist of the self-dual $\mathcal{N}=4$ case, this theory is mapped to the chiral algebra subsector of $\mathcal{N}=2$.

\section{Supersymmetric localization and matrix models}\label{sec:MM}

In this section, I consider matrix model subsectors of the boundary theories in our story and connect them to famous constructions of matrix models in the literature. 

\subsection{Supersymmetric localization on $S^4$}\label{subsec:MM-SL}

In his famous paper \cite{Pestun_2012}, Pestun computed $\mathcal{N}=4$ Wilson loops on $S^4$ using supersymmetric localization. Choosing a supercharge that squared to rotations, he localized the theory to a Gaussian matrix model on the north and south poles of the $S^4$ in such a way that Wilson loops on the equator were equal to a certain correlator in the matrix model.

There is a connection to the story we discussed in section \ref{sec:N=4}. If we consider the same setup on compactified twistor space $\CP^3$ (which is the twistor space of $S^4$), we can study the B-model on the Calabi-Yau
\begin{align*}
    X=\mathcal{O}(-1)^4\to\CP^3.
\end{align*}
We then consider a stack of D5 branes wrapping $\CP^3\subset X$. The theory on the brane is then the twistor uplift of self-dual $\mathcal{N}=4$ on the twistor space of $S^4$.

If we give $\CP^3$ the homogeneous coordinate $[Z_0:Z_1:Z_2:Z_3]$, then the analogous twist to the one performed in section \ref{sec:N=4} is concerned with the supercharge
\begin{align*}
    Q=Z_2\partial_{w_3}+Z_3\partial_{w_4}.
\end{align*}
This localizes the theory on the stack of branes to the chiral algebra subsector of $\mathcal{N}=4$ living on the twistor line at the origin of $S^4$. Choosing the supercharge $Q = Z_0\partial_{w_3}+Z_1\partial_{w_4}$ localizes the theory to the twistor line at the point at infinity on $S^4$. The chiral algebras living at the poles of the $S^4$ are illustrated in Fig. \ref{fig:twist_to_MM}.

\begin{figure}[h]
    \centering
    \includegraphics[width=0.95\linewidth]{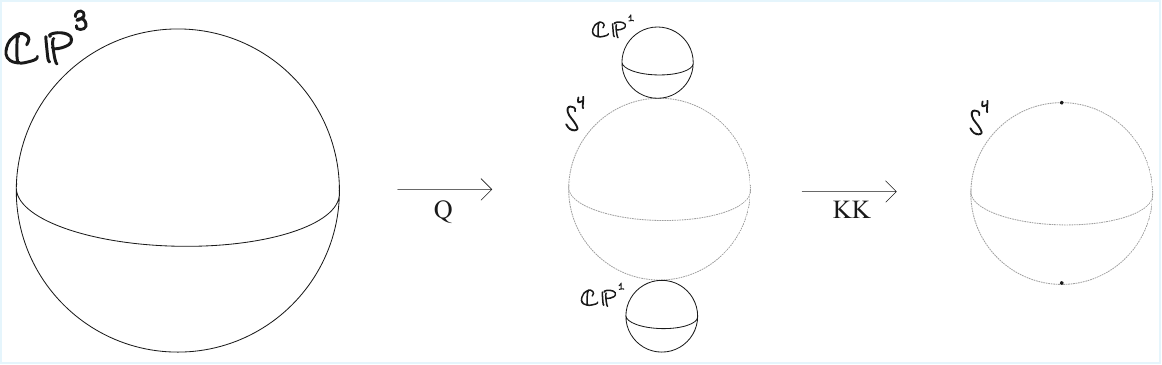}
    \caption{An illustration of the steps taken to arrive at the Pestun matrix model localization from the twistor uplift of self-dual $\mathcal{N}=4$. We start with the twistor uplift of self-dual $\mathcal{N}=4$ SYM, which we twist by the two choices of $Q$ to get two copies of the 2d chiral algebra subsector of $\mathcal{N}=4$ at the north and south poles of $S^4$. Then we KK compactify the 2d chiral algebras to arrive at Gaussian matrix models located at the north and south poles of $S^4$.}
    \label{fig:twist_to_MM}
\end{figure}

We thus arrive at a picture similar to the Pestun localization story, with the matrix models at the poles upgraded to 2d chiral algebras on Riemann spheres. Where the connection becomes explicit is in the matrix model subsector of the 2d chiral algebra found in \cite{Budzik_2023}. In \cite{Budzik_2023}, it was shown that one can compactify the chiral algebras to a Gaussian matrix model by only considering a specific subsector of states. One should also be able to obtain the zero-dimensional theory as a second twist of the 2d chiral algebra. This was done starting from the full supersymmetric theory in 4d in \cite{Drukker_2009}.

It would be nice if one could make this connection between the work of Pestun and this twistor space construction precise. I leave this for future work. It would also be interesting if one could compute the gravitational duals of the matrix model correlators considered by Pestun using this setup.

\subsection{The Dijkgraaf-Vafa matrix model}\label{subsec:MM-DV}

In the previous subsection, I argued that the Pestun localization on $S^4$ could be interpreted as the matrix model subsector of $\mathcal{N}=4$ found in \cite{Budzik_2023}. Here, I connect both stories to the Dijkgraaf-Vafa matrix model construction \cite{Dijkgraaf_2002}. My approach will be to show that the twisted holography setting studied in \cite{Budzik_2023,costello2021twistedholography} naturally emerges from the Dijkgraaf-Vafa construction.

In their paper, Dijkgraaf and Vafa \cite{Dijkgraaf_2002} study the B-model on the Calabi-Yau 3-fold\footnote{Notice that this is the same background that we landed on when localizing the bulk dual in the self-dual $\mathcal{N}=2$ story before the backreaction.}
\begin{align*}
    T^*\CP^1\times\C.
\end{align*}
Let's call the $\CP^1$ base coordinate $z$, $w_1$ on $\mathcal{O}(-2)$, and $w_2$ on $\C$.

\cite{Dijkgraaf_2002} wrap $N$ D1 branes around $\CP^1\subset T^*\CP^1\times\C$ and study the gauged $\beta\gamma$ system on the branes. Please see appendix \ref{app:MM-Beltrami} for a slow derivation of the brane action and for the naming conventions I use in this section. By introducing a Beltrami-differential of the form
\begin{align}\label{eq:MM-Beltrami}
    \mu = \frac{d\bar{z}}{(1+|z|^2)^2}\partial_{w_1}P(w_1)\partial_{w_2}
\end{align}
in the bulk geometry, for an arbitrary polynomial $P$, the theory on the brane becomes
\begin{align*}
    \int_{\CP^{1}}\Tr(\Phi_1\bar{D}\Phi_2+\omega\wedge P(\Phi_2)).
\end{align*}
Where $\bar{D}=\dbar+[A,\cdot]$ for a partial connection $A$. Dijkgraaf and Vafa then compactify this theory to the zero modes, resulting in the matrix model action
\begin{align*}
    P(\Phi_2)
\end{align*}
where $\Phi_2\in U(N)$ now denotes the constant zero mode. The special case where $P(w_1) = w_1^2$, is carefully worked out in appendix \ref{app:MM-Beltrami}. This case is important because it yields a Gaussian matrix model, which was also found in the 2d chiral algebra subsector of $\mathcal{N}=4$ \cite{Budzik_2023}.

Let me now connect these two stories. In the $\mathcal{N}=4$ case, the bulk geometry was
\begin{align*}
    \mathcal{O}(-1)^2\to\CP^1.
\end{align*}
Showing that the presence of the Beltrami given in equation \ref{eq:MM-Beltrami} with $P(w_1) = w_1^2$ deforms $T^*\CP^1\times\C$ to $\mathcal{O}(-1)^2\to\CP^1$ is sufficient. To understand why this is true, we can think of this choice of Beltrami differential as changing the geometry that Dijkgraaf and Vafa study to the one studied by \cite{Budzik_2023}. So, compactifying the brane theory in each setup is equivalent.

So, we conclude with
\begin{proposition}
    Introducing the Beltrami differential
    \begin{align*}
        \mu = 2\frac{d\bar{z}}{(1+|z|^2)^2}w_1\partial_{w_2}
    \end{align*}
    deforms the holomorphic vector bundle $T^*\CP^1\times\C$ to $\mathcal{O}(-1)^2\to\CP^1$.
\end{proposition}
\begin{proof}
    Recall that Beltrami differentials deform the complex structure by $\dbar\mapsto\dbar+\mu$. The holomorphic functions in this complex structure are then $z$, $w_1$, and
\begin{align*}
    \tilde{w}_2 = w_2+\frac{2w_1}{z(1+|z|^2)}.
\end{align*}
The holomorphic coordinates on our vector bundle are now $w_1,\tilde{w}_2$. On the other coordinate patch, our coordinates transform as
\begin{align*}
    z^2w_1\qquad w_2+\frac{2z^3|z|^2w_1}{1+|z|^{2}}.
\end{align*}
So, we can describe the deformed bundle by the transition map
\begin{align*}
    G = \begin{pmatrix}
        z^2&0\\
        \frac{2(z^3|z|^2-z^{1})}{1+|z|^{2}}&1
    \end{pmatrix}.
\end{align*}
In other words, our bundle has changed from being defined by $\text{diag}(z^2,1)$ to $G$. 

To see that this transition map defines $\mathcal{O}(-1)^2$, I will write the explicit gauge transformation. Call the bottom left entry of $G$, $g$. Then we can consider
\begin{align*}
    u_0 = \begin{pmatrix}
        1&0\\ g&1
    \end{pmatrix}\qquad u_\infty =\begin{pmatrix}
        1/z&0\\0&z
    \end{pmatrix}.
\end{align*}
It is then easy to see that $u_\infty Gu_0^{-1} = z\mathbb{1}$ which is the transition function defining $\mathcal{O}(-1)^2$.
\end{proof}

We thus see that turning on a specific Beltrami differential in the class studied in \cite{Dijkgraaf_2002} deforms the complex structure to that of the resolved conifold. A natural extension of this finding would be to understand how the backreacted geometry, $SL(2,\C)$, enters when we start with $T^*\CP^1\times\C$ in the presence of the Beltrami given in equation \ref{eq:MM-Beltrami}. One would expect to recover the findings of \cite{Budzik_2023} when performing a Dijkgraaf-Vafa type analysis of the matrix model.

\subsection*{Turning on bulk noncommutativity}

To wrap up the discussion, I give an alternative perspective on the Dijkgraaf-Vafa matrix model derivation. This is somewhat orthogonal to the main body of this subsection, but I leave it here as an interesting remark.

I interpret the Dijkgraaf-Vafa derivation of the matrix model action as introducing a bulk closed string field that turns on noncommutativity. Consider the bulk bivector
\begin{align*}
    \mu = \partial_{w_1}\partial_z.
\end{align*}
Introducing such a field corresponds to making the bulk geometry noncommutative \cite{Costello:2019jsy}. The dual supercharge is $Q = \theta_1\partial_z$. As we have seen in other examples, twisting by $Q$ renders the brane field independent of $z$, transforming it into a zero-dimensional field. The holomorphic Chern-Simons field can then be expanded in the surviving fermionic direction $\theta_2$ with constant coefficients. These coefficients are the field of our matrix model and its ghost with action given by the chosen polynomial $P$ in the Dijkgraaf-Vafa construction. 

\appendix

\section{Bochner-Martinelli manipulations}\label{app:BM_calc}

In sections \ref{sec:N=4} and \ref{sec:N=2}, I invert the $\dbar$-operator on certain Bochner-Martinelli type kernels. I perform some of these calculations explicitly here. First, we use
\begin{align*}
    \dbar&\left(\frac{\epsilon_{1ijk}\bar{w}^id\bar{w}^jd\bar{w}^k}{||w||^6} \right)  = 2\dbar\left(\frac{\bar{w}^2d\bar{w}^3d\bar{w}^4}{||w||^6} +\frac{\bar{w}^3d\bar{w}^4d\bar{w}^2}{||w||^6}+\frac{\bar{w}^4d\bar{w}^2d\bar{w}^3}{||w||^6}\right)\\
    & = -\frac{6}{||w||^8}(\bar{w}^2d\bar{w}^3d\bar{w}^4d\bar{w}^1w_1+\bar{w}^3d\bar{w}^4d\bar{w}^2d\bar{w}^1w_1+\bar{w}^4d\bar{w}^2d\bar{w}^3d\bar{w}^1w_1\\
    &\quad\quad +\bar{w}^2d\bar{w}^3d\bar{w}^4d\bar{w}^2w_2+\bar{w}^3d\bar{w}^4d\bar{w}^2d\bar{w}^3w_3+\bar{w}^4d\bar{w}^2d\bar{w}^3d\bar{w}^4w_4)+\frac{\epsilon_{1ijk}d\bar{w}^id\bar{w}^jd\bar{w}^k}{||w||^6}\\
    & = -\frac{6}{||w||^8}(\bar{w}^2d\bar{w}^3d\bar{w}^4d\bar{w}^1w_1+\bar{w}^3d\bar{w}^4d\bar{w}^2d\bar{w}^1w_1+\bar{w}^4d\bar{w}^2d\bar{w}^3d\bar{w}^1w_1-\bar{w}^1d\bar{w}^2d\bar{w}^3d\bar{w}^4w_1)\\
    & = 3\frac{\epsilon_{ijkl}\bar{w}^id\bar{w}^jd\bar{w}^kd\bar{w}^l}{||w||^8}w_1.
\end{align*}
A similar calculation holds when one fixes different indices on the $\epsilon$ tensor. Next, we make use of
\begin{align*}
    \dbar\left(\frac{\epsilon_{12ij}\bar{w}^id\bar{w}^j}{||w||^4} \right) & = \dbar\left(\frac{\bar{w}^3d\bar{w}^4-\bar{w}^4d\bar{w}^3}{||w||^4}\right)\\
    & = -\frac{2}{||w||^6}(\bar{w}^3w_1d\bar{w}^1d\bar{w}^4-\bar{w}^4w_1d\bar{w}^1d\bar{w}^3+\bar{w}^3w_2d\bar{w}^2d\bar{w}^4-\bar{w}^4w_2d\bar{w}^2d\bar{w}^3\\
    &\quad\quad+\bar{w}^3w_3d\bar{w}^3d\bar{w}^4-\bar{w}^4w_4d\bar{w}^4d\bar{w}^3)+\frac{d\bar{w}^3d\bar{w}^4-d\bar{w}^4d\bar{w}^3}{||w||^4}\\
    & = -\frac{2}{||w||^6}(\bar{w}^3w_1d\bar{w}^1d\bar{w}^4-\bar{w}^4w_1d\bar{w}^1d\bar{w}^3+\bar{w}^3w_2d\bar{w}^2d\bar{w}^4-\bar{w}^4w_2d\bar{w}^2d\bar{w}^3)\\
    & = 2\frac{\epsilon_{1ijk}\bar{w}^id\bar{w}^jd\bar{w}^k}{||w||^6}w_2+2\frac{\epsilon_{i2jk}\bar{w}^id\bar{w}^jd\bar{w}^k}{||w||^6}w_1.
\end{align*}
A similar identity holds when fixing different indices on the $\epsilon$ tensor.

\section{$SL(2,\C)/\Z_2$ calculations}\label{app:SL(2,C)/Z2calc}

In this appendix, I explicitly check the claims I made about holomorphic functions in the complex structure determined by $\bar{D}$ in section \ref{sec:N=2}.

\subsection{Verifying holomorphicity}

Holomorphicity of $f_1,f_2,f_3$ is trivial, so I focus on the other cases.

Recall that the new $\dbar$ operator was defined by
\begin{align*}
    \bar{D}=\dbar+ \frac{N}{4(2\pi)^3}\frac{ud\bar{u}}{t(1+|u|^2)^2}\partial_{v}.
\end{align*}
Only the $d\bar{u}$ direction of the derivative appears in these calculations, so I will not write the form part in the rest of this appendix. The holomorphicity of $f_4$ is then ensured by
\begin{align*}
    \bar{D}vt & = \frac{N}{4(2\pi)^3}\frac{u}{(1+|u|^2)^2}\qquad -\frac{N}{4(2\pi)^3}\bar{D}\frac{1+2|u|^2}{(1+u\bar{u})}=-\frac{N}{4(2\pi)^3}\frac{u}{(1+u\bar{u})^2}
\end{align*}
using the elementary derivative
\begin{align*}
    \partial_{\bar{u}}\frac{1+2|u|^2}{(1+u\bar{u})} = \frac{2u(1+|u|^2)-u(1+2|u|^2)}{(1+|u|^2)^2} = \frac{u}{(1+|u|^2)^2}.
\end{align*}
The holomorphicity of $f_5$ and $f_6$ are verified by an almost identical computation. We also have that $\bar{D}f_7 = 0$ since
\begin{align*}
    \bar{D}v^2t & = \frac{N}{4(2\pi)^3}\frac{2vu}{(1+|u|^2)^2}\\
    -\frac{N}{2(2\pi)^3}\bar{D}\frac{v(1+2|u|^2)}{(1+u\bar{u})}&=-\frac{N}{2(2\pi)^3}\frac{vu}{(1+u\bar{u})^2}-\frac{N^2}{8(2\pi)^6}\frac{u(1+2|u|^2)}{t(1+u\bar{u})^3}\\
    \frac{N^2}{16(2\pi)^6}\bar{D}\frac{(1+2|u|^2)^2}{t(1+|u|^2)^2} & = \frac{N^2}{16(2\pi)^6}\frac{2u(1+2|u|^2)}{t(1+|u|^2)^3}
\end{align*}
using 
\begin{align*}
    \partial_{\bar{u}}\frac{(1+2|u|^2)^2}{(1+|u|^2)^2} = \frac{4u(1+2|u|^2)(1+|u|^2)^2-2u(1+2|u|^2)^2(1+|u|^2)}{(1+|u|^2)^4} = \frac{2u(1+2|u|^2)}{(1+|u|^2)^3}.
\end{align*}
A similar calculation gives $\bar{D}f_8=\bar{D}f_9 = 0$.

\subsection{Verifying relations}
I found 6 relations displayed in equation \ref{eq:SL/Z2_rel}. The first and last relations are obvious. The second relation holds because
\begin{align*}
    f_8^2 & =v^4t^2u^2 -\frac{N}{2(2\pi)^3}\frac{v^3tu^2(1+2|u|^2)}{(1+u\bar{u})}+\frac{N^2}{16(2\pi)^6}\frac{v^2u^2(1+2|u|^2)^2}{(1+|u|^2)^2}\\
    &\quad-\frac{N}{2(2\pi)^3}\frac{v^3tu^2(1+2|u|^2)}{(1+u\bar{u})}+\frac{N^2}{4(2\pi)^6}\frac{v^2u^2(1+2|u|^2)^2}{(1+u\bar{u})^2}-\frac{N^3}{32(2\pi)^9}\frac{vu^2(1+2|u|^2)^2}{t(1+u\bar{u})^2}\\
    &\quad+\frac{N^2}{16(2\pi)^6}\frac{v^2u^2(1+2|u|^2)^2}{(1+|u|^2)^2} -\frac{N^3}{32(2\pi)^9}\frac{vu^2(1+2|u|^2)^3}{t(1+u\bar{u})^3}+\frac{N^4}{16^2(2\pi)^{12}}\frac{u^2(1+2|u|^2)^4}{t^2(1+|u|^2)^4}\\
    f_7f_9 & = v^4t^2u^2 -\frac{N}{2(2\pi)^3}\frac{v^3tu^2(1+2|u|^2)}{(1+u\bar{u})}+\frac{N^2}{16(2\pi)^6}\frac{v^2u^2(1+2|u|^2)^2}{(1+|u|^2)^2}\\
    &-\frac{N}{2(2\pi)^3}\frac{v^3tu^2(1+2|u|^2)}{(1+u\bar{u})} +\frac{N^2}{4(2\pi)^6}\frac{v^2u^2(1+2|u|^2)^2}{(1+u\bar{u})^2}-\frac{N^3}{32(2\pi)^9}\frac{vu^2(1+2|u|^2)^3}{t(1+u\bar{u})^3}\\
    &+\frac{N^2}{16(2\pi)^6}\frac{v^2u^2(1+2|u|^2)^2}{(1+|u|^2)^2} -\frac{N^3}{32(2\pi)^9}\frac{vu^2(1+2|u|^2)^3}{t(1+u\bar{u})^3}+\frac{N^4}{16^2(2\pi)^{12}}\frac{u^2(1+2|u|^2)^4}{t(1+|u|^2)^4}.
\end{align*}
For the third, fourth, and fifth relations, after computing the square terms, the match is trivial since $f_1,f_2,f_3$ are monomials. We see that
\begin{align*}
    f_4^2 & = v^2t^2-2\frac{N}{4(2\pi)^3}\frac{vt(1+2|u|^2)}{(1+u\bar{u})}+\frac{N^2}{16(2\pi)^6}\frac{(1+2|u|^2)^2}{(1+u\bar{u})^2}\\
    f_5^2 & = v^2t^2u^2-2\frac{N}{4(2\pi)^3}\frac{vtu^2(1+2|u|^2)}{(1+u\bar{u})}+\frac{N^2}{16(2\pi)^6}\frac{u^2(1+2|u|^2)^2}{(1+u\bar{u})^2}\\
    f_6^2 & = v^2t^2u^4-2\frac{N}{4(2\pi)^3}\frac{vtu^4(1+2|u|^2)}{(1+u\bar{u})}+\frac{N^2}{16(2\pi)^6}\frac{u^4(1+2|u|^2)^2}{(1+u\bar{u})^2}.
\end{align*}

\section{From Beltrami differentials to matrix models}\label{app:MM-Beltrami}

Here, I carefully explain how the deformation of the gauged $\beta\gamma$ system studied by Dijkgraaf and Vafa \cite{Dijkgraaf_2002} arises from introducing a Beltrami differential in the bulk using the topological string language of this paper. This is a basic computation that should come easily to those familiar with topological strings, but I add it here to keep this paper self-contained.

Recall that the theory on a D1 brane wrapping $\CP^1\subset T^*\CP^1\times\C$ is given by
\begin{align*}
    \int_{\CP^{1|2}} (\omega\lrcorner\partial_z) \wedge d\theta_1\wedge d\theta_2\wedge \text{hCS}(\mathcal{A})
\end{align*}
where $\mathcal{A}\in\Omega^{0,\bullet}(\CP^1,\g)[\theta_1,\theta_2]$, $\omega=dzd\bar{z}/(1+|z|^2)^2$ is the K\"ahler form on $\CP^1$, and $\theta_1,\theta_2$ are the fermionic coordinates dual to the $\mathcal{O}(-2)\times\C$ directions in the bulk.

We can gauge fix our superfield to remove the antiholomorphic form components and expand in terms of the fermionic directions
\begin{align*}
    \mathcal{A} = c+\Phi_i\theta^i+b\theta_1\theta_2
\end{align*}
where our fields are now sections of line bundles:
\begin{align*}
    c\in\mathcal{O}\qquad \Phi_1\in\mathcal{O}(-2)\qquad \Phi_2\in\mathcal{O}\qquad b\in\mathcal{O}(-2).
\end{align*}
We can use this expansion and integrate out the fermionic directions to get the action
\begin{align*}
    \int_{\CP^1}\Tr(\Phi_1\dbar\Phi_2+b\dbar c)
\end{align*}
where I have absorbed the $(1,0)$ form into $\Phi_1$ and $b$ since they are sections of the canonical bundle. This is the gauge fixed version of the action studied in \cite{Dijkgraaf_2002}.

Let me now explain how this theory is deformed by the Beltrami differential in the bulk. The Beltrami considered in this story takes the form
\begin{align*}
    P(w_1)\omega dw_1dw_2 \cong \frac{d\bar{z}}{(1+|z|^2)^2}\partial_{w_1}P(w_1)\partial_{w_2} = \mu
\end{align*}
using the isomorphism $PV^{\bullet,\bullet}(T^*\CP^1\times\C)\cong\Omega^{3-\bullet,\bullet}(T^*\CP^1\times\C)$. Turning on a Beltrami can be viewed as introducing a closed string field $\mu$ to the bulk, which, as explained in \cite{Costello:2019jsy}, deforms the theory on the brane by
\begin{align*}
    \int_{\CP^{1|2}}(\omega\lrcorner\partial_z) \wedge d\theta_1\wedge d\theta_2\wedge \Tr(\partial_{\theta_1}P(\theta_1)\mathcal{A}\partial_{\theta_2}\mathcal{A}).
\end{align*}
In section \ref{sec:MM}, I discuss the special case where $P(w_1) = w_1^2$. This leads to the following deformation of the brane theory
\begin{align*}
    2\int_{\CP^{1|2}}(\omega\lrcorner\partial_z) \wedge d\theta_1\wedge d\theta_2\wedge \Tr(\theta_1\mathcal{A}\partial_{\theta_2}\mathcal{A}) = 2\int_{\CP^{1|2}}(\omega\lrcorner\partial_z) \wedge d\theta_1\wedge d\theta_2\wedge \Tr(\Phi_2^2).
\end{align*}

\bibliographystyle{JHEP}

\bibliography{bib.bib}

\end{document}